\documentclass{article}
\usepackage{hyperref}
\usepackage{amssymb}
\usepackage{physics}
\usepackage[most]{tcolorbox}
\usepackage{titlesec}
\usepackage[T1]{fontenc}
\usepackage{amsthm}
\usepackage{amsmath}
\usepackage{mathrsfs}
\usepackage[noabbrev]{cleveref}
\usepackage{thmtools}
\usepackage{bbm}
\usepackage{bm}
\usepackage{xpatch}
\usepackage{float}
\usepackage{wrapfig}
\usepackage{subcaption}
\usepackage[labelfont=bf, labelsep=space]{caption}
\usepackage{sidenotes}
\tcbuselibrary{theorems}
\usepackage[Bjornstrup]{fncychap}
\usepackage{mdframed}
\usepackage[svgnames]{xcolor}
\usepackage{enumitem}
\usepackage{authblk}
\usepackage[
    backend=biber,
    style=authoryear,
    sorting=nyt,
    doi=true,
    url=true,
    giveninits=true,
    maxbibnames=99,
    maxcitenames=2,
    mincitenames=1
]{biblatex}
\usepackage{subfiles}

\newtheorem{theorem}{Theorem}
\newtheorem{proposition}{Proposition}
\newtheorem{lemma}{Lemma}
\newtheorem{corollary}{Corollary}

\title{Group Survival Probability under Contagion in Microlending}

\author{
    \begin{minipage}[t]{0.3\textwidth}
        \centering
        Héctor Jasso-Fuentes\\
        {\footnotesize
        Department of Mathematics\\ Cinvestav - Quéretaro\\Lib. Norponiente 2000, 76230 Juriquilla, Qro., Mexico \\ORCID 0000-0002-2502-9219\\hector.jasso@cinvestav.mx}
    \end{minipage}
    \hfill
    \begin{minipage}[t]{0.3\textwidth}
        \centering
        Alejandra Quintos\textsuperscript{1}\\
        {\footnotesize
        Department of Statistics\\University of Wisconsin-Madison\\Madison, WI 53706. USA\\ORCID 0000-0003-3447-3255\\alejandra.quintos@wisc.edu}
    \end{minipage}
    \hfill
    \begin{minipage}[t]{0.3\textwidth}
        \centering
        Xinta Yang\textsuperscript{2}\\
        {\footnotesize
        Department of Statistics\\
        University of Wisconsin-Madison\\
        Madison, WI 53706. USA\\
        ORCID 0009-0003-3959-3173\\
        xinta.yang@wisc.edu}
    \end{minipage}
}

\date{May 2025}

\DeclareNameAlias{default}{family-given}
\DeclareNameAlias{author}{family-given}
\DeclareNameAlias{sortname}{family-given}
\DeclareDelimFormat[bib]{multinamedelim}{\addcomma\space}
\DeclareDelimFormat[bib]{finalnamedelim}{\addcomma\space}
\DeclareDelimFormat[bib]{nametitledelim}{\addcolon\space} 
\DeclareDelimFormat[cite]{multinamedelim}{\addspace\bibstring{and}\space}
\DeclareDelimFormat[cite]{finalnamedelim}{\addspace\bibstring{and}\space}

\DeclareFieldFormat[article,inproceedings,book]{title}{#1}
\DeclareFieldFormat[article,inproceedings]{booktitle}{#1}
\DeclareFieldFormat[article]{journaltitle}{#1}
\DeclareFieldFormat[article,inproceedings]{volume}{#1}
\DeclareFieldFormat[article,inproceedings]{pages}{#1}

\DeclareBibliographyDriver{article}{%
  \printnames{author}%
  \setunit{\addcolon\space}
  \printfield[title]{title}%
  \setunit{\adddot\space}
  \printfield{journaltitle}%
  \setunit{\space}%
  \printfield{volume}%
  \setunit{\addcomma\space}%
  \printfield{pages}%
  \setunit{\space}%
  \printtext[parens]{\printfield{year}}%
  \finentry
}

\DeclareBibliographyDriver{book}{%
  \printnames{author}%
  \setunit{\addcolon\space}
  \printfield[title]{title}%
  \setunit{\adddot\space}%
  \printlist{publisher}%
  \setunit{\addcomma\space}%
  \printlist{location}%
  \setunit{\space}%
  \printtext[parens]{\printfield{year}}%
  \finentry
}

\DeclareBibliographyDriver{inproceedings}{%
  \printnames{author}%
  \setunit{\addcolon\space}
  \printfield[title]{title}%
  \setunit{\adddot\space}%
  \printtext{In:\space}%
  \printfield{booktitle}%
  \setunit{\addcomma\space}%
  \printfield{pages}%
  \setunit{\space}%
  \printtext[parens]{\printfield{year}}%
  \finentry
}

\begin{document}

\maketitle

\begin{abstract}
    In the context of microfinance, a group of individuals undertake business projects that may interfere with one another. A "contagious default" happens if one person's project failure leads to the default of another group member. In this paper, we apply a probabilistic approach to analyze the impact of such contagion among investment group members.
    Firstly, a general formula is provided to compute the group survival probability with the presence of contagion effect.
    Then, special cases of this probability model are examined in detail.
    In particular, we show that if the investment group is homogeneous, defined in the paper, then including more members into the group will eventually lead to default with probability $1$.
    This differs from the non-contagious scenario, where the default probability decreases monotonically with respect to the group size.
    Afterwards, we provide an upper bound of the optimal group size under the homogeneous setup; so, one can run a linear search within finite time to locate this optimizer.
\end{abstract}

\footnotetext[1]{Supported by the Office of the Vice Chancellor for Research and Graduate Education at the University of Wisconsin–Madison with funding from the Wisconsin Alumni Research Foundation.}
\footnotetext[2]{Corresponding author}

\section{Introduction}

The idea of microfinance appeared in the late-1970s when Mohammad Yunus began a pilot scheme lending money to villagers who had no access to conventional loans with the hope of reducing poverty and helping borrowers start their own business projects \parencite{yunus2007banker}.
Today, microfinance institutions are widespread in less developed regions where poor people face challenges borrowing money from banks due to a lack of credit history or collateral to back up their loans.
Despite the great success of micro-lending seen by the Grameen Bank founded by Yunus, many similar programs face challenges in collecting repayments from participants \parencite{adams1986rural}.
To increase the success rate of such financing programs, various means are implemented in practice.
First of all, instead of providing financial assistance to individuals, a group of people is often preferred. \textcite{conlin1999peer} presents a unique perspective explaining the success of peer-group micro-lending programs in north America, and \textcite{attanasio2015impacts} conducted randomized experiments in Mongolia showcasing the positive impact of group lending compared to individual lending. 
In addition, the concept of "joint liability" is often applied within the group, i.e. if an individual's business project fails, as long as others can help cover his or her repayment duty, the whole team is still considered safe and open to future re-financing opportunities.
The impact of this joint liability on repayment rates is analyzed by \textcite{huppi1990role} and \textcite{besley1995group}.
Also, one may intuitively agree that unless the individual risks are highly correlated, the overall risk of everybody, or the majority of the group, defaulting simultaneously may be low.
Secondly, the borrower group is usually self-selected in an effort to suppress free-riding, as a close social ties may enhance peer pressure.
Whether or not the presence of such a social tie is effective in improving repayment rates remains debatable \parencite{besley1995group, wydick1999can}.
Thirdly, different penalty schemes in the event of default are incorporated. 
One of the stricter rules is known as contingent renewal, which limits or denies future loans to all group members if any one of them defaults \parencite{chowdhury2007group}.
For more discussion on the topic of defaulting penalties, refer to the works by \textcite{diener2009mathematical} and \textcite{diener2016randomness}.

Besides the many existing literature that study the different lending mechanisms hoping to increase the repayment rate as mentioned in the previous paragraph, the problem of optimizing the group size has been explored by economists such as \textcite{ahlin2015role}, \textcite{rezaei2017optimal}, \textcite{abbink2006group}; and mathematicians such as \textcite{protter2022quintos}.
We include in our model the concept of group contagion, where a member who defaults purely due to his or her own project setbacks would generate a negative effect, which may bring other people down with a certain probability.
The idea of contagion was also studied in the works of \textcite{davis2001infectious} and \textcite{erol2023contagion}. Our paper differs from these studies in the sense that we use contagion to analyze the optimal group size, and our computation is based on a probabilistic approach instead of conducting social experiments or performing a graph-theoretical analysis.

The outline of this paper is as follows. 
In \cref{section:3}, we present the group model and our definition of group default.
The main body is in \cref{section:4}, where we compute the group survival probability in terms of group size and the contagion factor.
Afterwards in its subsection, a special case is analyzed in detail, where we prove that unless there is absolutely no contagion, simply including more participants will eventually lead to the group default.
Although the exact optimal group size cannot be obtained analytically, we provide, towards the end of this section, an explicit upper bound for where this optimal value lies, hence a numerical approach is readily available.

\section{Notation and Convention} \label{section:2}

This section presents the notations and conventions used throughout the paper for ease of reference.

\begin{itemize}
    \item Probability related notations
    \begin{itemize}
        \item We fix a complete probability space $(\Omega, \mathscr{F}, \mathbb{P})$.
        
        \item Given an index set $I$, $\eta_{i\in I}$ refers to the random vector $(\eta_i)_{i\in I}$, and $\eta_{i\in I} \le T$ means $\eta_i \le T$ for all $i\in I$.
        
        \item For a given random variable $X$,  the notation $X\in\mathscr{F}$ means $X$ is $\mathscr{F}-$ measurable.
    \end{itemize}

    \item Computational conventions
    \begin{itemize}
        \item Products over an empty set are defined as $0$, i.e. $\prod_{i\in I} a_i = 1$, when $I = \emptyset$.

        \item $0^0 = 1$

        \item $\mathcal{I}^n_k$ denotes the collection of all index subsets of size $k$ out of $n$ elements. For example, if $n = 3$ and $k = 2$, then $\mathcal{I}^3_2 = \{\{1,2\}, \{1, 3\}, \{2, 3\}\}$.
    \end{itemize}
\end{itemize}

\section{Model Setup} \label{section:3}

Suppose we have $n + 1$ individuals forming a group that plans to borrow money for their business projects, where $n$ of them are regular group members, and there is a group leader.
Define the default time of the i-th member due solely to their own investment failure (i.e., independent of others) as $\eta_i$, which is a random variable measurable with respect to $\mathscr{F}$. We index individuals by $i = 0, ..., n$, where $i = 0$ refers to the group leader.

We now introduce the concept of a contagion factor among \textit{regular} group members. 
If member $i$ defaults due to their own investment failure, they might negatively affect other members in the group. Let $\{Y_{ij} \in \mathscr{F}\}_{i, j = 1, ..., n; i\ne j}$ be binary-valued random variables indicating whether member $i$ causes member $j$ to default. 

In this framework, $\eta_i$ represents the time at which member $i$ would default in the absence of contagion, while $Y_{ij}$ captures the peer effect of contagion among regular group members. These two components are modeled separately to distinguish between individual risk and group-driven influence. To ensure analytical tractability under our framework presented later, we impose the following assumption. 

\begin{enumerate}[label=(A\arabic*)]
    \item Group member $i$ may influence others via contagion only if their own default is caused by their own investment failure, not due to contagion from others.
\end{enumerate}
If we do not assume (A1), then there will be chains of contagion among the regular group members, and because group members form a network, our analysis below will quickly become intractable. 
Also, we assume that if the group leader (indexed by $i=0$) defaults, then the entire group is considered to have defaulted. Thus, we only work with the contagion factor among the $n$ regular group members.

Under this framework, the group is considered to have defaulted if and only if one of the following 2 events happens: (i) the group leader defaults, or (ii) all regular members default, regardless of the leader's status. 

Our goal is to compute the survival probability of the group 
\begin{equation} \label{equation:group-default-def}
    \mathbb{P}(\text{the group still survives at time }T)
\end{equation}
where $T$ represents the end of the financing period.

\subsection{Contagion Factor}
The motivation for introducing the idea of contagion is that if one person defaults, their failure might negatively influence others. The influence should naturally depend on time: given that member $i$ defaults at time $\eta_i$, the probability that they cause member $j$ to default before the terminal time $T$ should be a function of $\eta_i$. Intuitively, if $\eta_i$ is small, there is more time for the negative effect to impact $j$, thus the corresponding probability should be bigger.
Formally, we represent the contagion effect by binary random variables $Y_{ij}$, $i, j = 1, ..., n$ and $i \ne j$, satisfying the following conditional distribution, 
\begin{equation} \label{equation:contagion-1}
    \begin{aligned}
        \mathbb{P}(Y_{ij} = 1 \,|\, \eta_1, ..., \eta_n) &:= q_{ij}(\eta_i) \\
        \mathbb{P}(Y_{ij} = 0 \,|\, \eta_1, ..., \eta_n) &:= 1 - q_{ij}(\eta_i)
    \end{aligned}
\end{equation}
where $q_{ij}(t)$ is a non-increasing function of time. Since these are probabilities, we also require $q_{ij}(t) \ge 0$, and we may assume $q_{ij}(t) = 0$ for $t > T$, as contagion is irrelevant after the investment period ends.
When $Y_{ij} = 1$, it means member $i$ causes member $j$ to default before the terminal time $T$.

A remark is in order.
The definition above emphasizes that the distribution of $Y_{ij}$ depends \textit{only} on the default time $\eta_i$ of member $i$, not on other members' default times.
Formally, this implies
\begin{align*}
    \mathbb{P}(Y_{ij} = 1 | \eta_i) &= \mathbb{E}[\mathbb{P}(Y_{ij} = 1 | \eta_1, ..., \eta_n) | \eta_i] \\
    &= \mathbb{E}[q_{ij}(\eta_i) | \eta_i] \\
    &= q_{ij}(\eta_i) = \mathbb{P}(Y_{ij} = 1 | \eta_1, ..., \eta_n)
\end{align*}

This setup implies that, given all members' natural default times, whether member $i$ brings down member $j$ through contagion depends solely on $\eta_i$.

In addition, we make two further assumptions.
\begin{enumerate}[label=({A}\arabic*)]
    \setcounter{enumi}{1}
    \item The $\eta_i \text{'s are mutually independent}$
    \item The $Y_{ij} \text{'s are conditionally independent of each other given } (\eta_1, ..., \eta_n)$
\end{enumerate}
Assumption (A2) reflects the idea that each individual's investment is independent of the others. Assumption (A3) requires further justification.
First, $Y_{ij}$ models interactions between members $i$ and $j$, while $Y_{ik}$ concerns $i$ and $k$; so they do not necessarily affect one another.
Also, symmetry between $Y_{ij}$ and $Y_{ji}$ is not required, since each arises from a different member's project.
Second, the variables $\{Y_{ij}\}$ and $\{\eta_i\}$ model different aspects of the system (group interactions versus individual risk) so we assume their dependence structure is separate.

At first glance, these independence assumptions might suggest the group behaves like a collection of independent members.
However, the inclusion of contagion variables $\{Y_{ij}\}$ introduces meaningful dependencies.
For instance, consider the probability that member $i$ defaults before time $T$:
\begin{align*}
    &\mathbb{P}(i\text{ defaults before } T) = \mathbb{P}(\eta_i \le T) \\
    &\hspace{7em}+ \sum_{J\subset \{1, ..., n\}}\mathbb{P}(\eta_i > T, \eta_{j\in J} \le T, Y_{ji} = 1 \text{ for at least one } j\in J)
\end{align*}
This depends on all $\{\eta_j\}_{j=1}^n$, not just $\eta_i$; so the individual default probabilities are not independent.
This confirms that the model captures more than just a collection of non-interacting individuals.

Now, we provide two plots of $q_{ij}(t)$ below, which illustrate two special cases: 
\begin{itemize}
    \item "immediate constant contagion"
    \item "constant contagion up to time $t=T-\Delta$ and no contagion afterward"
\end{itemize}
\begin{figure}[h!]
    \centering
    \begin{subfigure}[b]{0.45\linewidth}
        \includegraphics[width=\linewidth]{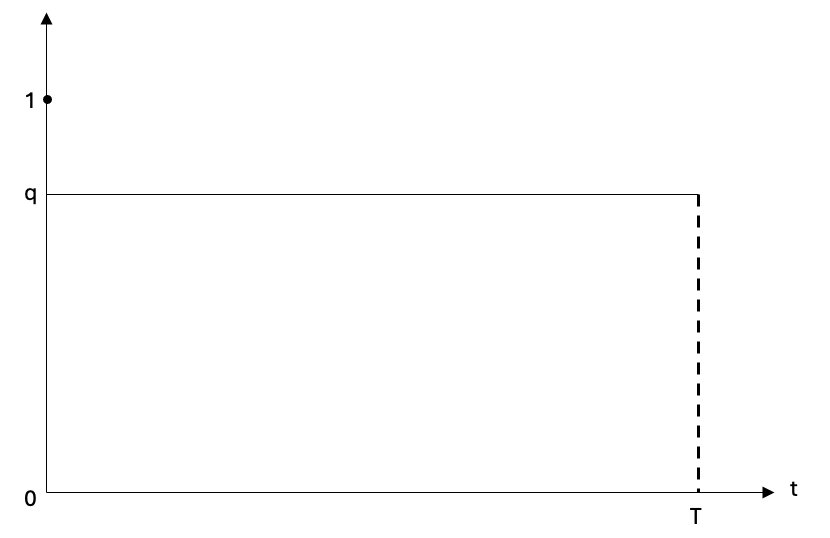}
        \caption{Immediate contagion with constant probability $q$}
    \end{subfigure}
    \begin{subfigure}[b]{0.45\linewidth}
        \includegraphics[width=\linewidth]{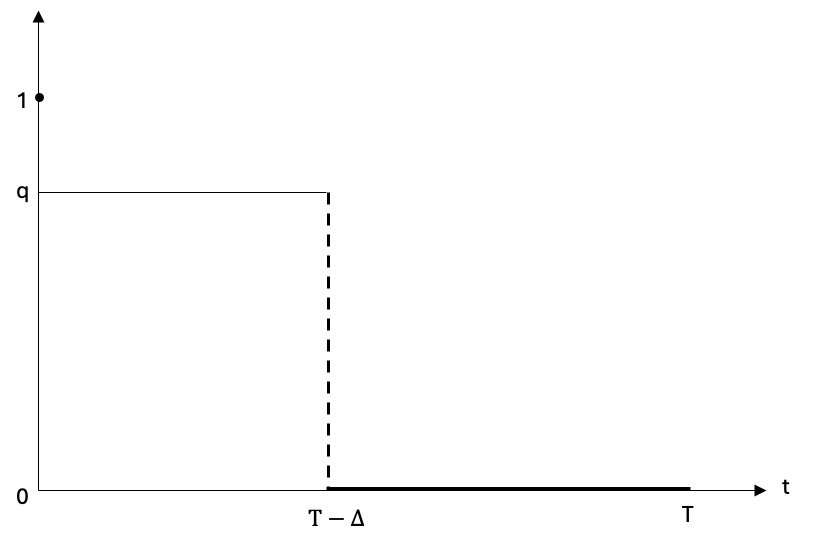}
        \caption{Contagion with probability $q$ if the "natural defaulting" time is early enough}
    \end{subfigure}
    \caption{Two examples of the function $q_{ij}(t)$}
    \label{figure:q_ij}
\end{figure}
In the first example, as soon as member $i$ defaults, he/she introduces a positive, constant probability of bringing down member $j$, which also persists throughout the whole investment period $T$. 
In the second example, member $i$ must default early enough in order for contagion to occur, specifically before the threshold $T - \Delta$.
This second model means contagion needs time to propagate to others.
If member $i$ defaults too late, there is insufficient time for the negative influence to reach other regular members before time $T$, which is represented by the zero value in the plot after $T - \Delta$.

\section{Main Results} \label{section:4}

To compute the group survival probability under this contagion framework, we first state and prove a useful lemma.

\begin{lemma} \label{lemma:group-default-prob-with-natural-default-index-set-I}
    Given an index set $I \subset \{1, 2, ..., n\}$, define the event $A_I = \{\eta_{i\in I} \le T, \eta_{j\in I^c} > T, \text{ all regular members default by } T\}$, where notation $\eta_{i\in I} \le T$ means $\eta_i \le T$ for all $i\in I$ as mentioned in the notation section above.
    We then have 
    \begin{equation} \label{equation:group-default-prob-with-natural-default-index-set-I-eq-1}
            \mathbb{P}(A_I) = \mathbb{E}\left\{ \mathbbm{1}(\eta_{i\in I} \le T)\cdot \prod_{j\in I^c} \left(\mathbb{P}(\eta_j > T)\cdot \left[1 - \prod_{i\in I} (1 - q_{ij}(\eta_i)) \right] \right) \right\}
    \end{equation}
    If $I = \{1, ..., n\}$, the right-side of \cref{equation:group-default-prob-with-natural-default-index-set-I-eq-1} simplifies to $\mathbb{P}(\eta_1, ..., \eta_n \le T)$, consistent with the convention for products over the empty set.
\end{lemma}
\begin{proof}
    We proceed by first conditioning on $\eta_{i\in\{1, ..., n\}}$, which refers to the random vector $(\eta_1, .., \eta_n)$ as mentioned before.
    \begin{align*}
        \mathbb{P}(A_I) &= \mathbb{E}\left[\vphantom{\prod_{j\in I^c}}\mathbbm{1}(\eta_{i\in I} \le T)\cdot\mathbbm{1}(\eta_{j\in I^c} > T)\cdot \right.\\
        &\hspace{5em}\left. \mathbb{E}\left[\prod_{j\in I^c}\mathbbm{1}(Y_{ij} = 1 \text{ for some }i\in I) \,\bigg|\, \eta_{i\in\{1, ..., n\}}\right]\right]
    \end{align*}
    This is because a member $j\in I^c$ can only default due to contagion, which means some member $i\in I$ needs to trigger $Y_{ij} = 1$.
    For the inner conditional expectation above, we continue as follows.
    \begin{align*}
        &\quad\,\, \mathbb{E}\left[\prod_{j\in I^c}\mathbbm{1}(Y_{ij} = 1 \text{ for some }i\in I) \,\bigg|\, \eta_{i\in\{1, ..., n\}}\right] \\
        &= \mathbb{E}\left[\prod_{j\in I^c} 1 - \mathbbm{1}(Y_{ij} = 0 \text{ for all }i\in I) \,\bigg|\, \eta_{i\in \{1, ..., n\}}\right]
    \end{align*}
    By using repeatedly assumption (A3), the last expression turns out to be 
    \begin{align*}
        &= \prod_{j\in I^c} \mathbb{E}\left[1 - \mathbbm{1}(Y_{ij} = 0 \text{ for all }i\in I) \,|\, \eta_{i\in \{1, ..., n\}} \right] \\
        &= \prod_{j\in I^c} 1 - \mathbb{P}\left( Y_{ij} = 0 \text{ for all } i \in I \,|\, \eta_{i\in\{1, ..., n\}} \right) \\
        &= \prod_{j\in I^c} \left(1 - \prod_{i\in I} \mathbb{P}\left[Y_{ij} = 0 \,|\, \eta_{i\in\{1, ..., n\}}\right]\right) = \prod_{j\in I^c} \left(1 - \prod_{i\in I}[1-q_{ij}(\eta_i)] \right)
    \end{align*}
    Also, notice the above quantity belongs to the sigma-algebra generated by $\eta_{i\in I}$ denoted by $\sigma(\eta_{i\in I})$.
    Substitute this result back into the first equation in the proof, 
    \begin{align*}
        \mathbb{P}(A_I) &= \mathbb{E}\left[\mathbbm{1}(\eta_{i\in I} \le T)\cdot \mathbbm{1}(\eta_{j\in I^c} > T)\cdot \prod_{j\in I^c} \left(1 - \prod_{i\in I}[1 - q_{ij}(\eta_i)] \right)\right] \\
        &= \mathbb{E}\left[\mathbbm{1}(\eta_{i\in I} \le T)\cdot \prod_{j\in I^c} \left(1 - \prod_{i\in I}[1 - q_{ij}(\eta_i)] \right)\cdot \mathbb{E}[\mathbbm{1}(\eta_{j\in I^c} > T) \,|\, \eta_{i\in I}] \right] \\
        &= \mathbb{E}\left[\mathbbm{1}(\eta_{i\in I} \le T)\cdot \prod_{j\in I^c} \left(1 - \prod_{i\in I}[1 - q_{ij}(\eta_i)] \right)\cdot \prod_{j\in I^c}\mathbb{P}(\eta_j > T)\right]
    \end{align*}
    where the last step is because all $\eta_i$'s are independent by assumption (A2).
\end{proof}

\begin{theorem} \label{theorem:general-group-surviving-rate}
    The group survival probability at time $T$ is 
    $$\mathbb{P}(\text{group still survives by time } T) = \mathbb{P}(\eta_0 > T)\cdot \left(1 - \sum_{k=1}^n \sum_{I \in \mathcal{I}^n_k} \mathbb{P}(A_I) \right)$$
    where the event $A_I$ and $\mathcal{I}^n_k$ are defined in \Cref{lemma:group-default-prob-with-natural-default-index-set-I} and \Cref{section:2} respectively.
\end{theorem}
\begin{proof}
    Recall the definition of ``group default" right above \cref{equation:group-default-def}, we compute 
    \begin{align*}
            &\quad\,\, \mathbb{P}(\text{group still survives by } T) \\ 
            &= 1 - \mathbb{P}(\text{leader defaults}) - \mathbb{P}(\text{leader stands, all regular members default}) \\
            &= 1 - \mathbb{P}(\eta_0 \le T) - \mathbb{P}(\eta_0 > T)\cdot \mathbb{P}(\text{all regular members default}) \\
            &= \mathbb{P}(\eta_0 > T)\cdot [1 - \mathbb{P}(\text{all regular members default})]
    \end{align*}
    Now, the probability that all regular members default is the sum over all possible ways this can happen: some subset $I \subset \{1, ..., n\}$ defaults naturally, while the rest in $I^c$ default via contagion.
    Each such scenario corresponds to the event $A_I$, and the events $A_I$ are disjoint across $I$. Hence, using the expression for $\mathbb{P}(A_I)$ derived in \Cref{lemma:group-default-prob-with-natural-default-index-set-I}, we have
    $$\mathbb{P}(\text{all regular members default}) = \sum_{k=1}^n\sum_{I \in \mathcal{I}_k^n} \mathbb{P}(A_I)$$
    leading to the desired result.
\end{proof}

With the relation established in the theorem above, we now use the next two corollaries to illustrate two special cases of $q_{ij}(t)$. The first special case is called ``constant" contagion because once someone defaults, there is a constant positive probability of bringing down others. The second case occurs when contagion requires an early enough default time. After presenting these two results, we then explore some properties of the ``constant" contagion scenario in more detail.

\begin{corollary}[Immediate Constant Contagion]
\label{corollary:immediate-contagion}
    Suppose the contagion function is constant: $q_{ij}(t) \equiv q_{ij}\cdot \mathbbm{1}(t \le T)$ for $t \ge 0$ (refer to \cref{figure:q_ij}(a)); then 
    \begin{equation} \label{equation:immediate-contagion-eq-1}
        \mathbb{P}(\text{group survives at time }T) = \mathbb{P}(\eta_0 > T)\cdot \left(1 - \sum_{k=1}^n \sum_{I\in \mathcal{I}^n_k} \beta_I\gamma_I \right)
    \end{equation}
    where 
    \begin{equation} \label{equation:immediate-contagion-eq-2}
    \begin{aligned}
        \beta_{I} &= \prod_{i\in I} \mathbb{P}(\eta_i \le T) \\
        \gamma_{I} &=
        \prod_{j\in I^c} \left\{ \mathbb{P}(\eta_j > T) \cdot \left[ 1-\prod_{i \in I} (1-q_{ij}) \right] \right\}
    \end{aligned}
    \end{equation}
    (If $I = \{1, ..., n\}$, then $\gamma_I := 1$ by convention, as mentioned in \Cref{section:2}.)
\end{corollary}
\begin{proof}
    In virtue of \Cref{lemma:group-default-prob-with-natural-default-index-set-I} and the hypothesis of $q_{ij}(t) = q_{ij}\mathbbm{1}(t \le T)$, we can easily declare the following:
    \begin{align*}
        \mathbb{P}(A_I) &= \mathbb{P}(\eta_{i\in I} \le T)\cdot \prod_{j\in I^c} \left( \mathbb{P}(\eta_j > T)\cdot \left[1 - \prod_{i\in I} (1-q_{ij}) \right] \right) \\
        &= \prod_{i\in I} \mathbb{P}(\eta_{i} \le T)\cdot \prod_{j \in I^c} \left(\mathbb{P}(\eta_j > T)\cdot \left[1 - \prod_{i\in I} (1-q_{ij}) \right] \right) \\
        &= \beta_I\gamma_I
    \end{align*}
    Now, the proof concludes by invoking \Cref{theorem:general-group-surviving-rate}.
\end{proof}

We now present the second special case mentioned above, where contagion occurs only if members default early enough.

\begin{corollary}[Contagion requires early default]
\label{corollary:delayed-contagion}
    Suppose contagion only occurs when members default early, i.e. before time $T - \Delta$: $q_{ij}(t) = q_{ij}\cdot \mathbbm{1}(t \le T - \Delta)$ for $t \ge 0$ (refer to \cref{figure:q_ij}(b)), then the group survival probability is
    \begin{equation} \label{equation:delayed-contagion-eq-1}
        \mathbb{P}(\text{group survives by } T) = \mathbb{P}(\eta_0 > T)\cdot \left(
            1 - \sum_{k=1}^{n} \sum_{I \in \mathcal{I}^n_k} \sum_{\substack{I'\subset I \\ I'\ne \emptyset}} \beta_{I'} \cdot \gamma_{I, I'} \cdot \tau_{I, I'} \right)
    \end{equation}
    where 
    \begin{equation} \label{equation:delayed-contagion-eq-2}
        \begin{aligned}
            \beta_{I'} &= \prod_{i \in I'} \mathbb{P}(\eta_i \le T-\Delta) \\
            \gamma_{I, I'} &= \prod_{i \in I\setminus I'} \mathbb{P}(\eta_i \in (T-\Delta, T]) \\
            \tau_{I, I'} &= \prod_{j\in I^c} \left( \mathbb{P}(\eta_j > T)\cdot \left[ 1 - \prod_{i\in I'} (1-q_{ij}) \right] \right)
            \end{aligned}
    \end{equation}
    As mentioned in \Cref{section:2}, by convention, if $I = \{1, ..., n\}$, then $\tau_{I, I'} = 1$ and if $I = I'$, $\gamma_{I, I'} = 1$.
\end{corollary}
\begin{proof}
    Since the contagion factor $q_{ij}(t) = 0$ after $t > T - \Delta$, we can further split up the naturally defaulting group $I$ into two subsets: one subset $I'$ defaults early enough, i.e. $\eta_i \le T - \Delta$, and the other subset $I\setminus I'$ defaults too late for contagion to take place, i.e. $\eta_i \in (T - \Delta, T]$.
    To be more precise, given the naturally defaulting index set $I$, write 
    $$ \{\eta_{i\in I} \le T\} = \bigcup_{I'\subset I} \big\{\eta_{i\in I'} \le T-\Delta, \eta_{i\in I\setminus I'} \in (T-\Delta, T] \big\} $$
    where $I'$ ranges over all possible subsets of $I$, and we emphasize that the right side of the above equation is a disjoint union. For notational simplicity, we name each event $\{\eta_{i\in I'} \le T - \Delta, \eta_{i\in I\setminus I'} \in (T-\Delta, T]\}$ as $B_{I, I'}$, so that $\{\eta_{i\in I} \le T\} = \cup_{I' \subset I} B_{I,I'}$.
    
    Use similar reasoning as in \Cref{lemma:group-default-prob-with-natural-default-index-set-I}, for each event $B_{I, I'}$, we have
    \begin{align*}
        &\quad\,\, \mathbb{P}(B_{I,I'}, \eta_{j\in I^c} > T,  \text{ all regular members default}) \\
        &= \mathbb{E}\left[ \mathbbm{1}_{B_{I, I'}}\cdot \prod_{j\in I^c} \left(\mathbb{P}(\eta_j > T)\cdot \left[1 - \prod_{i\in I'} (1-q_{ij}) \right] \right) \right] \\
        &= \prod_{j\in I^c} \left(\mathbb{P}(\eta_j > T)\cdot \left[1 - \prod_{i\in I'} (1-q_{ij}) \right] \right) \cdot \mathbb{E}\left[\mathbbm{1}_{B_{I,I'}}\right] \\
        &= \prod_{j\in I^c} \left(\mathbb{P}(\eta_j > T)\cdot \left[1 - \prod_{i\in I'} (1-q_{ij}) \right] \right) \cdot \\
        &\hspace{7em}\prod_{i\in I'} \mathbb{P}(\eta_i \le T - \Delta) \cdot \prod_{i\in I\setminus I'} \mathbb{P}(\eta_i \in (T-\Delta, T])
    \end{align*}
    Notice the product involving $q_{ij}$'s is taken over the set $I' \subset I$, because in this setting, only those who default early enough can potentially bring others down.
    The rest of the proof is summing over all index sets $I \subset \{1, ..., n\}$ and subsets $I' \subset I$, a process similar to the one in \Cref{theorem:general-group-surviving-rate}.
\end{proof}

\noindent Remarks:
\begin{enumerate}[label=(\alph*)]
    \item The triple summation term in \cref{equation:delayed-contagion-eq-1} corresponds to the following process: first, select $k$ members out of $n$ to default on their own; then, among them, select a subset $I'$ to represent those who default on their own before time $T-\Delta$. This leaves the members in $I\setminus I'$ to default naturally during the interval $(T-\Delta, T]$.

    \item If $\Delta = 0$, we return to the constant contagion case and claim that \cref{equation:delayed-contagion-eq-1} agrees with \cref{equation:immediate-contagion-eq-1}. Indeed, in \cref{equation:delayed-contagion-eq-1}, when $I' \subsetneq I$, $\gamma_{I, I'} = 0$ because $\mathbb{P}(\eta_i \in (T-0, T]) = \mathbb{P}(\eta_i \in \emptyset) = 0$; and when $I' = I$, $\gamma_{I, I'} = 1$ by convention. So, the triple summation reduces to a double summation. Moreover, when $I = \{1, ..., n\}$, we have $\tau_{I, I'} = 1$ for any $I' \subset I$.
    
    \item On the other hand, if $\Delta > T$, then defaults can never trigger contagion, i.e. contagion is not possible. In that case, every term in the triple summation vanishes, since $\beta_{I'} = 0$ for all $I'$.
    
    \item In summary, as we vary the value of $\Delta$, \Cref{corollary:immediate-contagion} emerges as a special case of \Cref{corollary:delayed-contagion}.
\end{enumerate}

\subsection{Special Cases for Constant Contagion}

In this subsection, we revisit the constant contagion case and explore a few important properties.
Suppose that contagion is uniform across the group, i.e. $q_{ij} = q$ for all $i, j = 1,..., n$ with $i\ne j$. 
Under this assumption, the expression for $\gamma_{I}$ in \cref{equation:immediate-contagion-eq-2} simplifies to 
\begin{equation} \label{equation:homogeneous-group-immediate-contagion-eq-1}
    \gamma_{I} =  \left[ \prod_{j\in I^c} \mathbb{P}(\eta_j > T) \right] \cdot \left[ 1 - (1-q)^{|I|} \right]^{n-|I|}
\end{equation}
where $|I|$ is the number of members who default naturally.
Substituting this into \cref{equation:immediate-contagion-eq-1}, we see that the group survival probability is monotone decreasing in $q$.
That is, the group is more likely to survive when the contagion probability $q$ is smaller.
This aligns with intuition: the higher the contagion intensity, the more likely other members are to be brought down when a member defaults.

Next, let us examine the boundary cases $q = 0$ and $q = 1$. 
Recall \Cref{theorem:general-group-surviving-rate} and substitute the simplified expression for $\gamma_I$ from \cref{equation:homogeneous-group-immediate-contagion-eq-1} into \cref{equation:immediate-contagion-eq-1}.
This yields:
\begin{align} \label{equation: reference-of-theorem-1}
    &\quad\,\, \mathbb{P}(\text{group still survives by } T) \nonumber \\
    &= \mathbb{P}(\eta_0 > T)\cdot \left(1 - \sum_{k=1}^n \sum_{I \in \mathcal{I}^n_k} \left\{\prod_{i\in I} \mathbb{P}(\eta_i \le T) \right\}\cdot \right. \nonumber\\
    &\hspace{12em} \left.\left\{\left[\prod_{j\in I^c} \mathbb{P}(\eta_j > T)\right]\cdot \left[1 - (1-q)^k\right]^{n-k} \right\} \right)
\end{align}
Substituting $q = 0$ into \cref{equation:homogeneous-group-immediate-contagion-eq-1}, and observe that
$[1 - 1^{|I|}]^{n-|I|} = \mathbbm{1}(|I| = n)$ by our convention from \Cref{section:2}, namely $0^0 = 1$. 
The first multiplication term $\prod_{j\in I^c} \mathbb{P}(\eta_i > T)$ of $\gamma_I$ equals $1$ if $|I| = n$ by our convention on the multiplication over an empty set. Combining the two terms, we have
$\gamma_I = \mathbbm{1}(|I| = n)$.
Plugging this into \cref{equation: reference-of-theorem-1} yields
\begin{equation} \label{equation:homogeneous-group-immediate-contagion-eq-2}
    \mathbb{P}(\text{group still survives by time }T) = \mathbb{P}(\eta_0 > T)\cdot \left[ 1 - \prod_{i=1}^n \mathbb{P}(\eta_i \le T) \right]
\end{equation}
Next, set $q = 1$ in \cref{equation:homogeneous-group-immediate-contagion-eq-1}.
In this case, for any non-empty set $I$, $\gamma_I = \prod_{j\in I^c} \mathbb{P}(\eta_j > T)$.
Plugging this into \cref{equation: reference-of-theorem-1} gives
\begin{align} \label{equation:homogeneous-group-immediate-contagion-eq-3}
    &\quad\,\, \mathbb{P}(\text{group still survives by time }T) \nonumber \\
    &= \mathbb{P}(\eta_0 > T)\cdot \left(1 - \sum_{k=1}^n \sum_{I \in \mathcal{I}^n_k} \left[ \prod_{i\in I} \mathbb{P}(\eta_i \le T) \right] \cdot \left[ \prod_{j\in I^c} \mathbb{P}(\eta_j > T) \right]\right) \nonumber \\
    &= \mathbb{P}(\eta_0 > T)\cdot \left(1 - \sum_{k=1}^n \sum_{I \in \mathcal{I}^n_k} \mathbb{P}(\eta_{i\in I} \le T, \eta_{j \in I^c} > T) \right) \nonumber \\
    &= \mathbb{P}(\eta_0 > T)\cdot \left(1 - \mathbb{P}\left(\cup_{i=1}^n \{\eta_i \le T\}\right) \right) \nonumber \\
    &= \mathbb{P}(\eta_0 > T)\cdot \prod_{i=1}^n \mathbb{P}(\eta_i > T) = \prod_{i=0}^n \mathbb{P}(\eta_i > T)
\end{align}

Actually one can directly compute \cref{equation:homogeneous-group-immediate-contagion-eq-2} and \cref{equation:homogeneous-group-immediate-contagion-eq-3} as follows.
When $q=0$, all members act independently; 
the group survives as long as at least one regular member remains by $T$.
When $q = 1$, the default of any member cause the entire group to default.
Since $\gamma_I$ is monotone decreasing in $q$, 
\cref{equation:homogeneous-group-immediate-contagion-eq-2} gives the upper bound and \cref{equation:homogeneous-group-immediate-contagion-eq-3} the lower bound for the group survival probability of under the assumption of $q_{ij}(t)=q$. 

Now, suppose we \textit{further} assume regular group members have identical performance, i.e. $\{\eta_i\}_{i = 1, ..., n}$ are identically distributed.
This could occur, for example if all regular group members perform the same type of investment.
Then, \cref{equation: reference-of-theorem-1} simplifies to 
\begin{equation} \label{equation:homogeneous-group-immediate-contagion-eq-4}
    \mathbb{P}(\eta_0 > T) \cdot \left\{ 1 - \sum_{k=1}^n {n\choose k} \cdot \mathbb{P}(\eta_1 \le T)^k \cdot \mathbb{P}(\eta_1 > T)^{n-k} \cdot \left[ 1 - (1-q)^k \right]^{n-k} \right\}
\end{equation}
for $q \in (0, 1)$.
This assumption is reasonable in situations where there is a single investment plan, but the organizer wishes to spread the risk by having several people execute it.
Without contagion, increasing the number of members always reduces the chance that all regular members default before $T$,
unless $\mathbb{P}(\eta_i \le T)$ depends on the group size $n$ as analyzed by \textcite{protter2022quintos}.
However, when contagion is present, our experimental computations show that there exists an optimal group size, and the group's survival probability converges to $0$ as group size tends to infinity.
This observation motivates the next subsection, where we analyze \cref{equation:homogeneous-group-immediate-contagion-eq-4} in detail to understand how the optimal group size depends on the contagion factor $q$ and the individual survival probability $\mathbb{P}(\eta_1 \le T)$.

\subsection{Optimal Group Size under Constant Contagion and Homogeneity}
For notational simplicity, set $c_1 := \mathbb{P}(\eta_1 \le T), c_2 := \mathbb{P}(\eta_1 > T) = 1 - c_1$ and $c_3 := 1 - q$. The summation term from \cref{equation:homogeneous-group-immediate-contagion-eq-4} becomes
\begin{equation} \label{equation:homogeneous-group-immediate-contagion-eq-5}
    S_n := \sum_{k=1}^n {n\choose k}\cdot c_1^k\cdot c_2^{n-k}\cdot (1 - c_3^k)^{n-k}
\end{equation}
Before studying how $n$ affects $S_n$, we briefly examine the roles of the other two parameters, namely the individual's survival probability $c_2$ and 1 minus the contagion factor, i.e. $c_3$.

Fix $c_2$ and let $c_3$ vary from $0$ to $1$.
For $k < n$, $(1 - c_3^k)^{n-k}$ decreases with $c_3$, while the $k=n$ term is constant;
thus $S_n$ is non-increasing in $c_3$ with $S_n = 1 - c_2^n$ when $c_3 = 0$ and $S_n = (1 - c_2)^n$ when $c_3 = 1$.
Now, suppose we have two scenarios, characterized by $(c_2, c_3)$ and $(c_2', c_3')$, with $c_3 > c_3'$. 
This means the first scenario has a weaker contagion effect.
Suppose $c_3$ is close to 1 and $c_3'$ is close to 0.
Because $S_n$ is continuous with respect to $c_3$, it is close to $c_1^n = (1-c_2)^n$ in the first case, and close to $1-c_2'^n$ in the second case.
It is possible to have $c_2 < c_2'$ but still $(1 - c_2)^n < 1 - c_2'^n$.
For example take $c_2 = 0.2$, $c_2' = 0.9$, $n = 5$.
This shows that when contagion factor is high, even if each individual is performing well (characterized by the higher $c_2'$ in the example above), the whole group may still perform poorly (characterized by the higher $S_n$ value.)
From this discussion, we see that the contagion factor $q$ can play a decisive role in group performance, sometimes outweighing differences in individual survival probability.

In the next proposition and theorem, we study the effect of group size on the group's surviving probability when $0 < c_1, c_2, c_3 < 1$. The boundary cases $c_1 = 0, c_1 = 1$ are excluded as they represent no natural default at all and guaranteed natural default respectively.
Figure \ref{figure:surviving-rate-against-group-size-with-diff-q} plots the probability of the group surviving, i.e. \cref{equation:homogeneous-group-immediate-contagion-eq-4}, under different contagion factors $q$, with the group size $n$ on the horizontal axis. For illustration, we fix $\mathbb{P}(\eta_1 > T) = 0.3$ and $\mathbb{P}(\eta_0 > T) = 0.5$ in all four graphs. 
\begin{figure}[h!]
    \centering
    \begin{subfigure}[b]{0.475\linewidth}
        \includegraphics[width=\linewidth]{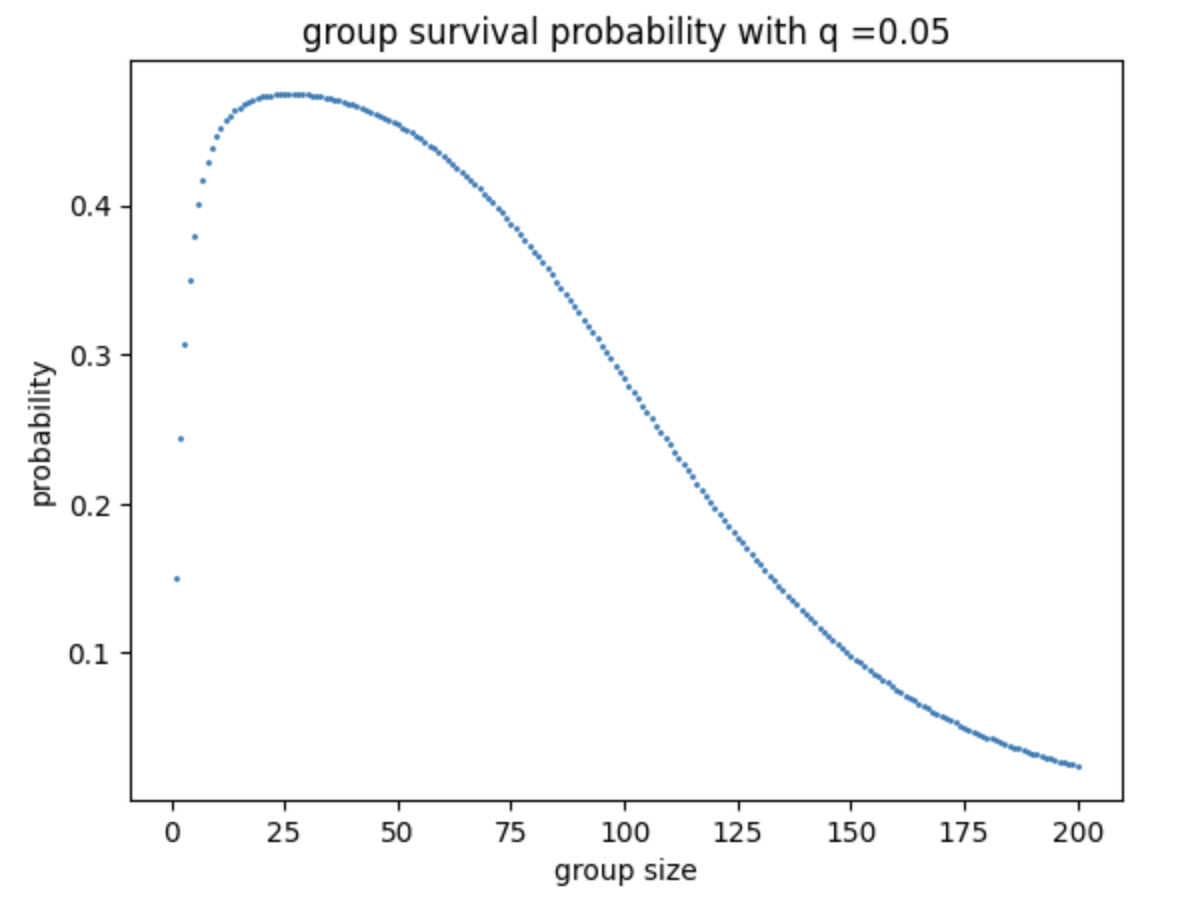}
        \caption{Homogeneous contagion factor $q = 0.05$}
    \end{subfigure}
    \begin{subfigure}[b]{0.475\linewidth}
        \includegraphics[width=\linewidth]{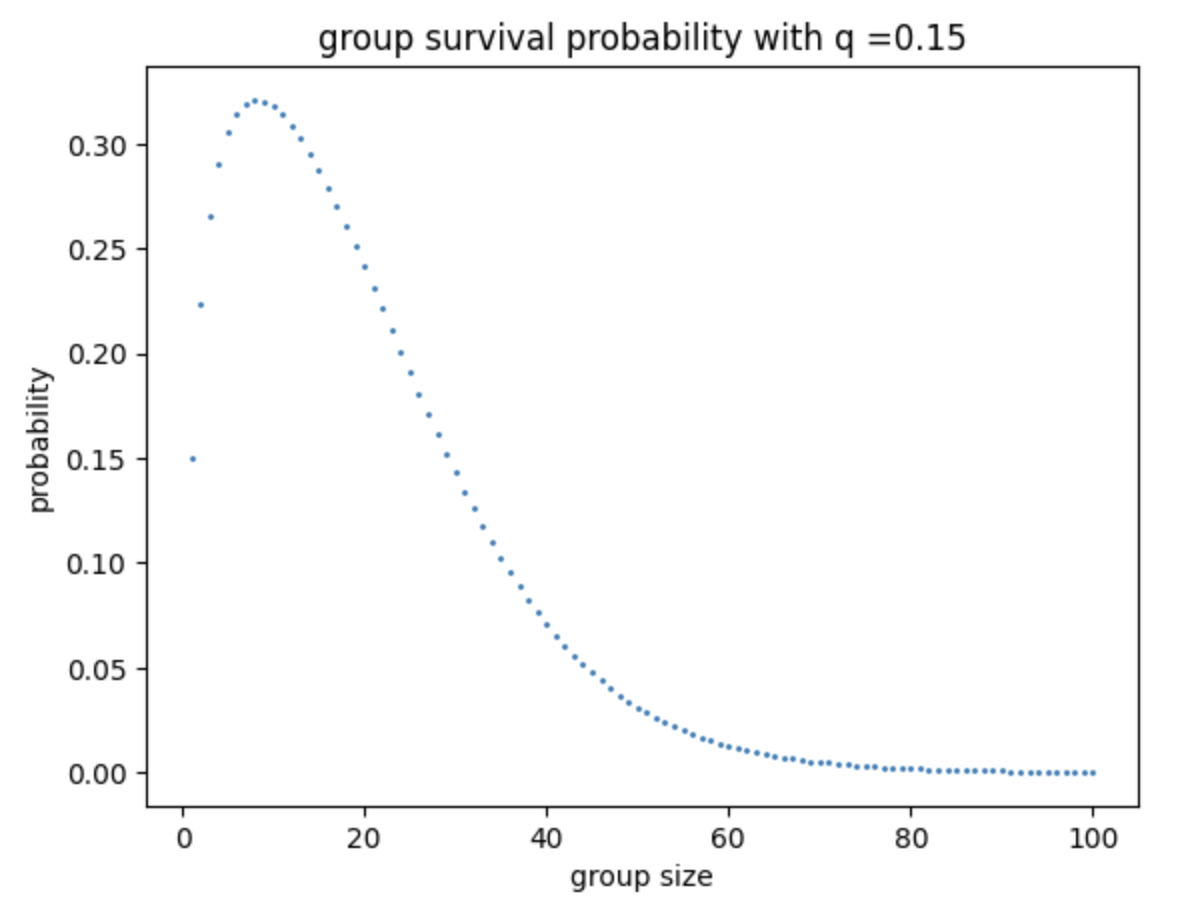}
        \caption{Homogeneous contagion factor $q = 0.15$}
    \end{subfigure}
    \begin{subfigure}[b]{0.475\linewidth}
        \includegraphics[width=\linewidth]{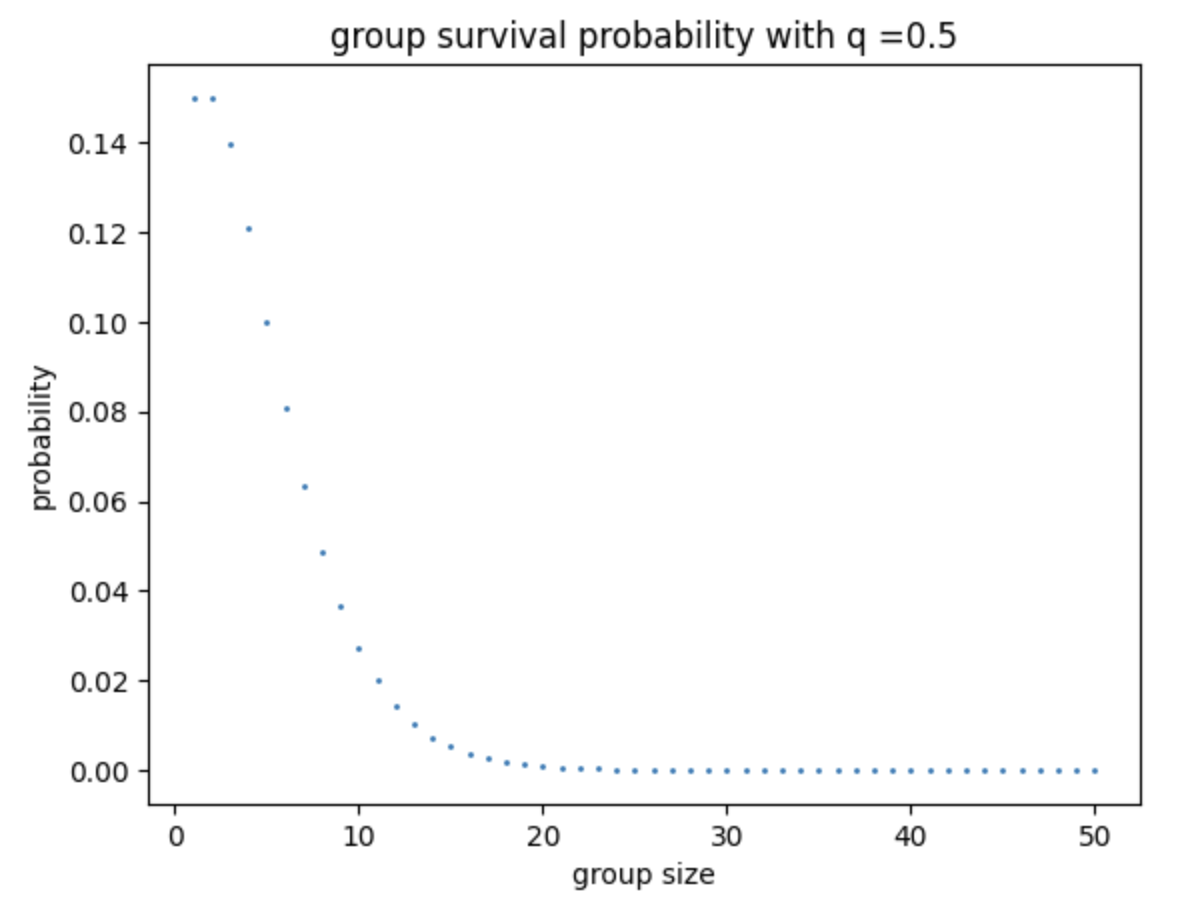}
        \caption{Homogeneous contagion factor $q = 0.5$}
    \end{subfigure}
    \begin{subfigure}[b]{0.475\linewidth}
        \includegraphics[width=\linewidth]{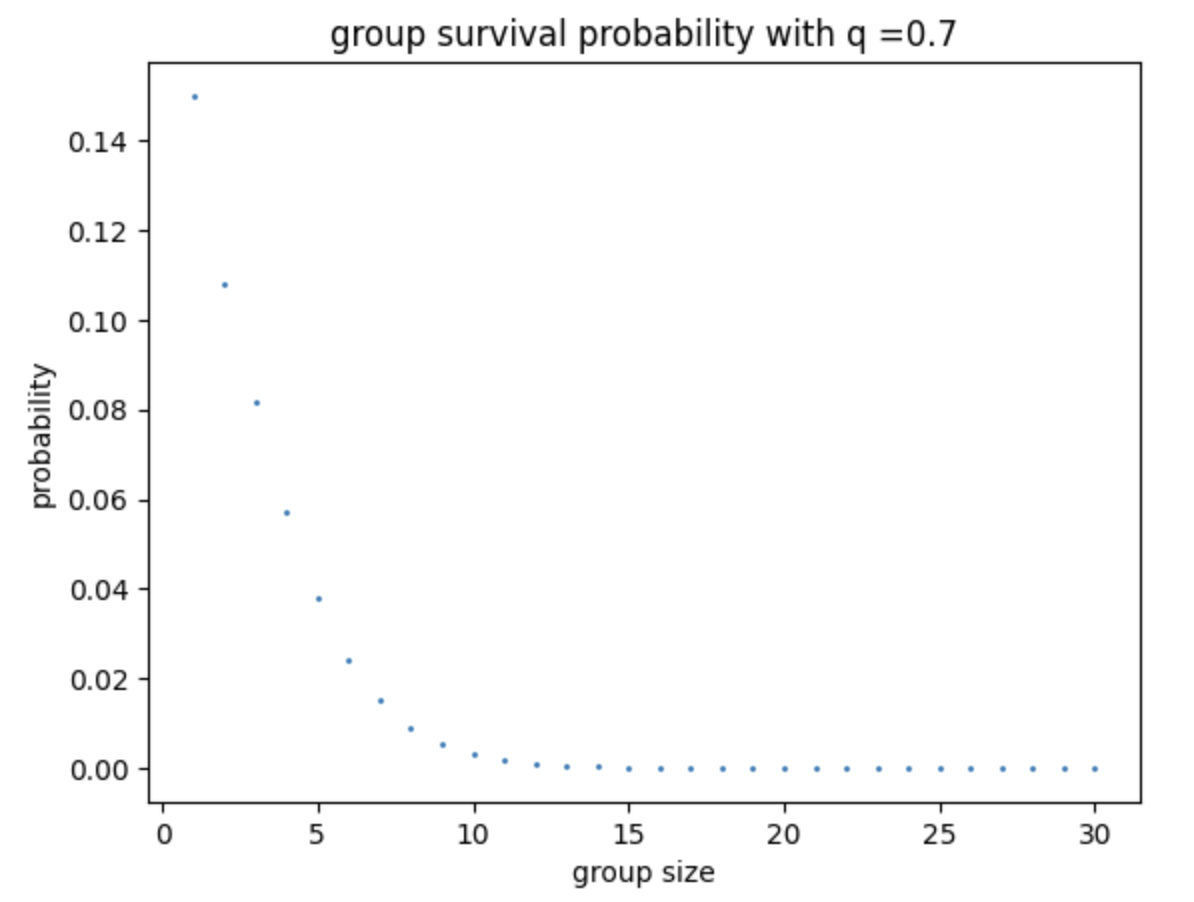}
        \caption{Homogeneous contagion factor $q = 0.7$}
    \end{subfigure}
    \caption{Probability of the group surviving as a function of group size $n$}
    \label{figure:surviving-rate-against-group-size-with-diff-q}
\end{figure}

The plots show that when $q$ is small, survival probability initially increases with $n$ before eventually approaching $0$; whereas for larger $q$ the decline is rapid, and the optimal number of regular members may be just $1$.
In other words, with low contagion, adding members can initially improve survival by diversifying the risk; but with high contagion, this advantage disappears quickly.

Now, we formalize these observations and characterize the range of the optimal group size.

\begin{theorem} \label{theorem:homogeneous-group-Sn-converges}
    Under the assumption of $q_{ij}(t) = q$ and that $\{\eta_i\}_{i = 1, ..., n}$ have the same distribution, i.e. the group is homogeneous, the quantity $S_n$ defined in \cref{equation:homogeneous-group-immediate-contagion-eq-5} satisfies $1 - S_n \le 3\cdot a^n$ for some $a\in(0, 1)$, when $n$ is sufficiently large.
\end{theorem}
\begin{proof}
    First notice we can change the starting value of $k$ to $0$ in $S_n$ because when $k$ takes $0$, the corresponding summand term is $0$. 
    Thus, we can write $S_n = \sum_{k=0}^n {n\choose k}\cdot c_1^k\cdot c_2^{n-k}\cdot (1-c_3^k)^{n-k}$.
    Let $\{Y_i\}_{i\in\mathbb{N}}$ be i.i.d. $\text{Bernoulli}(c_1)$ variables, and define
    $X_n = \sum_{i=1}^n Y_i$.
    So, $X_n \sim \text{Bin}(n, c_1)$, and we can express $S_n$ as
    \begin{equation} \label{equation:homogeneous-group-Sn-converges-eq-1}
        S_n = \mathbb{E}\left[\left(1 - c_3^{X_n} \right)^{n-X_n}\right]
    \end{equation}
    For any $\epsilon > 0$, using Hoeffding's concentration inequality, we have 
    \begin{equation} \label{equation:homogeneous-group-Sn-converges-eq-2}
        \mathbb{P}\left(\left|\frac{X_n}{n} - c_1 \right| \ge \epsilon \right) = \mathbb{P}(|X_n - nc_1| \ge n\epsilon) \le 2\exp\left(-\frac{2n^2\epsilon^2}{n} \right) = 2\exp(-2n\epsilon^2)
    \end{equation}
    Pick $\epsilon = c_1/2$.
    Define event $A_n = \{|X_n/n - c_1| \ge c_1/2\}$ and write $Z_n := \left(1 - c_3^{X_n} \right)^{n - X_n}$.
    By conditional expectation,
    \begin{equation} \label{equation:homogeneous-group-Sn-converges-eq-3}
        1 - S_n = \underbrace{\mathbb{E}\left[1 - Z_n \,|\, A_n\right] \cdot \mathbb{P}(A_n)}_{I_1} + \underbrace{\mathbb{E}\left[1 - Z_n \,|\, A_n^c\right] \cdot \mathbb{P}(A_n^c)}_{I_2}
    \end{equation}
    Because $Z_n \in (0, 1]$, $I_1 \le \mathbb{P}(A_n) \le 2\exp(-nc_1^2/2)$.
    Next, in the event $A_n^c$, we have $X_n > nc_1/2$.
    Since $c_3 \in (0, 1)$, monotonicity implies that 
    $$Z_n = \left(1 - c_3^{X_n} \right)^{n - X_n} \ge \left(1 - c_3^{nc_1/2} \right)^{n(1-c_1/2)}$$
    on $A_n^c$.
    Now, we bound $I_2$. 
    \begin{equation} \label{equation:homogeneous-group-Sn-converges-eq-4}
        \begin{aligned}
            I_2 \le \mathbb{E}[1 - Z_n \,|\, A_n^c] &\le \mathbb{E}\left[1 - \left(1 - c_3^{nc_1/2} \right)^{n(1 - c_1/2)} \,\bigg|\, A_n^c\right] \\
            &= 1 - \left(1 - c_3^{nc_1/2} \right)^{n(1-c_1/2)}
        \end{aligned}
    \end{equation}
    Use the binomial inequality, i.e. $(1-x)^n \ge 1 - nx$ for $x\in (0, 1)$, to get
    \begin{equation} \label{equation:homogeneous-group-Sn-converges-eq-5}
        1 - \left(1 - c_3^{nc_1/2} \right)^{n(1-c_1/2)} \le 1 - \left(1 - nc_3^{nc_1/2} \right)^{1-c_1/2}
    \end{equation}
    Next, we claim there exists some $N$ such that when $n > N$, we have $nc_3^{nc_1/2} < nc_3^{nc_1/6} < 1$.
    Indeed, $nc_3^{nc_1/6} < 1$ is equivalent to $\frac{\log n}{n} < \frac{c_1}{6}\log \frac{1}{c_3}$, and we choose $N$ based on $\frac{\log n}{n} \rightarrow 0$ as $n \rightarrow \infty$.
    Now, for $n > N$, we further upper bound \cref{equation:homogeneous-group-Sn-converges-eq-5} via 
    \begin{equation} \label{equation:homogeneous-group-Sn-converges-eq-6}
        \text{RHS of \cref{equation:homogeneous-group-Sn-converges-eq-5}} \le 1 - \left(1 - nc_3^{nc_1/2} \right) = nc_3^{nc_1/2}
    \end{equation}
    because $nc_3^{nc_1/2}, c_1/2 \in (0, 1)$.
    Note that with $n > N$, we further have
    \begin{equation} \label{equation:homogeneous-group-Sn-converges-eq-7}
        nc_3^{nc_1/2} < c_3^{nc_1/3}
    \end{equation}
    Indeed, the previous inequality is equivalent to $nc_3^{nc_1/6} < 1$, which is the criterion for our choice of $N$.
    Finally, combining the above bounds yields
    \begin{align} \label{equation:homogeneous-group-Sn-converges-eq-8}
        I_1 &\le 2\exp(-nc_1^2/2) = 2[\exp(-c_1^2/2)]^n & \forall n \nonumber\\
        I_2 &\le \left(c_3^{c_1/3} \right)^n & n > N
    \end{align}
    Setting $a = \max\left(\exp(-c_1^2/2), c_3^{c_1/3} \right) \in (0, 1)$, \cref{equation:homogeneous-group-Sn-converges-eq-3} and \cref{equation:homogeneous-group-Sn-converges-eq-8} complete the proof.
\end{proof}

The implication of this theorem is two-fold. First, it shows that $S_n \rightarrow 1$ as $n\rightarrow \infty$ when a contagion effect exists. In other words, the group eventually defaults as we include more and more people (see equations (\ref{equation:homogeneous-group-immediate-contagion-eq-4}) and (\ref{equation:homogeneous-group-immediate-contagion-eq-5})).
Second, $1 - S_n$ decays exponentially fast as $n$ grows, meaning that the group's chance of surviving becomes very slim before $n$ is very large. As shown in \cref{figure:surviving-rate-against-group-size-with-diff-q}, the optimal group size is usually not large. One observation from the plots is that when $q$ is small, i.e. a weak contagion, the optimal group size tends to be larger. This aligns with the fact that when there is no contagion, the optimal group size is infinity (see \cref{equation:homogeneous-group-immediate-contagion-eq-2}).

The theorem therefore indicates that if contagion is present, there exists a \textit{finite} optimal group size that maximizes the group's survival probability. The following proposition shows that if contagion is small, this optimal group size involves strictly more than one \textit{regular} member.

\begin{proposition}\label{prop:more-than-one-person}
    If the contagion factor $q < 1/2$, then the optimal group size is strictly larger than $1$.
\end{proposition}
\begin{proof}
    We explicitly compute $S_1$ and $S_2$ from \cref{equation:homogeneous-group-immediate-contagion-eq-5}.
    Recalling that $c_2 = 1 - c_1$, we have
    $$S_1 = c_1, \quad S_2 = c_1(2 - 2c_3 + 2c_1c_3 - c_1),$$
    which gives $$S_2 - S_1 = c_1(c_1 - 1)(2c_3 - 1).$$
    
    So, $c_3 = 1/2$ is the boundary case.
    Since $c_3 = 1 - q$, the condition $q < \frac{1}{2}$ is equivalent to $c_3 > \frac{1}{2}$, which yields $S_2 < S_1$.
    This means that a group with two regular members has a strictly higher survival probability by time $T$ than a "group" with only one regular member.
\end{proof}

Finally, we conclude our discussion by suggesting an upper bound for the optimal group size.
The main theorem above tells us only that the optimal group size is never infinite as long as contagion exists.
Before proceeding, we mention that if one substitute $X_n$ by its large number approximation $nc_1$ i.e. its expectation, we obtain an explicit optimizer for the group size $$n^* = \frac{\log(1/2)}{c_1\log c_3}$$ (as the group size is an integer, the true optimizer will be one of $n^*$'s neighboring integers.) Because this point deviates from the main discussion of the paper, we only briefly sketch the justification below.

Denote by $T_n$ the estimated value of $S_n$, i.e. $$T_n := \left(1 - c_3^{nc_1}\right)^{n - nc_1}.$$
To find the minimizer of $T_n$, it is equivalent to minimize $$f(n) := (n - nc_1)\log(1 - c_3^{nc_1}).$$
We temporarily treat $n \in \mathbb{R}$.
Compute $f'(n)$ and set it to zero.
Let $d = c_3^{n^*c_1} \in (0, 1)$, where $n^*$ is the optimizer of $f(n)$.
Then $d$ satisfies $$(1 - d)\log(1-d) = d\log d.$$
A unique solution $d = 1/2$ exists (one can check this by studying the convexity of function $x\log x$).
So, we obtain the claimed $n^*$.

However, the actual value of $S_n$ may differ noticeably from its estimate $T_n$, especially when $n$ is small, which means that $n^*$ may deviate from the true optimizer of $S_n$. 
Although an explicit formula for the optimal group size may not be readily available, the following result establishes an upper bound for it.

\begin{proposition} \label{prop: optimal-size-bound}
    Optimal group size $n$ is less than $\max\left(N, \frac{\log \frac{1-c_1}{3}}{\log a} \right)$, where 
    $$N = \begin{cases}
        1 & \frac{c_1}{6}\log \frac{1}{c_3} > \frac{1}{e} \\
        -\frac{6W_{-1}\left(-\frac{c_1|\log c_3|}{6} \right)}{c_1|\log c_3|} & \text{otherwise}.
    \end{cases}$$
    and $W_{-1}$ is the secondary branch of the Lambert W function, and $a$ is defined in \Cref{theorem:homogeneous-group-Sn-converges}.
\end{proposition}
\begin{proof}
    We first make precise the choice of $N$ used in the proof of \Cref{theorem:homogeneous-group-Sn-converges}.
    Recall that $N$ must satisfy 
    \begin{equation} \label{equation:N-criterion}
        \frac{\log n}{n} < \frac{c_1}{6}\log \frac{1}{c_3}
    \end{equation}
    for any $n > N$.
    For notational simplicity, write $\alpha = \frac{c_1}{6}\log \frac{1}{c_3}$.
    Comparing $f(x) = \log x$ and $g(x) = \alpha x$, $N$ can be taken as the larger root of $\log x = \alpha x$, if such root exists.
    A standard calculus computation shows that if $\alpha > 1/e$, the two functions $f(x), g(x)$ never intersect. In this case, $N=1$ suffices.

    If $\alpha < 1/e$, then $\log x = \alpha x$ has two solutions.
    This is because the function $h(x) = \frac{\log x}{x}$ attains its maximum value $1/e$ at $x = e$.
    Hence, the horizontal line at height $\alpha$ intersects the graph of $h(x)$ twice, producing two roots for the equation $\log x = \alpha x$.
    For our purpose, we take the larger one since we need \cref{equation:N-criterion} to hold for all large $n$.
    To express this root, rewrite the equation as $-\alpha x = -\alpha e^{\alpha x}$.
    Setting $z = -\alpha x$ gives $ze^z = -\alpha$.
    The larger $x$ corresponds to the smaller $z$ since they have opposite sign.
    Thus we select $z = W_{-1}(-\alpha)$, which then gives $$x = -\frac{W_{-1}(-\alpha)}{\alpha}$$
    where $W_{-1}$ denotes the secondary branch of the Lambert $W$ function
    (\textcite{corless1996lambert} provides more details on why $-W_{-1}$ gives the larger root.)

    Next, recall from \Cref{theorem:homogeneous-group-Sn-converges} that $1 - S_n \le 3a^n$ for some $a \in (0, 1)$ and $\forall n > N$.
    Define $x'$ by requiring $1 - S_1 = 3a^{x'}$.
    Since $a \in (0, 1)$, it follows that for any $n > x'$, $$3a^n < 3a^{x'} = 1 - S_1.$$
    Hence for such $n$, 
    $$1 - S_n \le 3a^n < 1 - S_1,$$
    so $S_n > S_1$.
    As the group survival probability equals $\mathbb{P}(\eta_0 > T)\cdot (1 - S_n)$ (see equations (\ref{equation:homogeneous-group-immediate-contagion-eq-4}) and (\ref{equation:homogeneous-group-immediate-contagion-eq-5})), it follows that beyond this threshold $\max(x', N)$, the survival probability of an $n$-regular-member group is strictly smaller than that of the 1-regular member group.
    The conclusion now follows.
\end{proof}

The value of $N$ above, although always finite, can be quite large. In the next proposition, we introduce the idea of a suboptimal group size attempting to reduce $N$.

We say $n$ is $\delta$-suboptimal if the group's survival probability with $n$ regular members is greater than its true optimal survival probability minus $\delta$. 
In other words, the value of $\mathbb{P}(\eta_0 > T)\cdot (1 - S_n)$ at the sub-optimal value $n$ is within $\delta$ of its true maximum.
We note that in many cases of $(c_1, c_3)$ considered in our numerical experiments, the upper bound obtained below is indeed smaller than the one presented in the previous proposition.

\begin{proposition} \label{prop:suboptimal-minimizer}
    Set $b = \max\left(e^{-c_1^2/2}, c_3^{c_1/2} \right) \in (0, 1)$. Then for any $\delta > 0$, there exists a $\delta$-suboptimal group size $n$ satisfying $n \le \lceil\log_b\frac{\delta(1-b)}{2} \rceil$.
\end{proposition}

\begin{proof}
    We bound $S_{n+1}$ in term of $S_n$, where we continue to use the notation from before with $X_n = \sum_{i = 1}^n Y_i$ denoting the number of natural defaults among $n$ regular members.
    First, 
    \begin{align} \label{equation:suboptimal-minimizer-eq-1}
        S_{n+1} &= \mathbb{E}\left[\left(1 - c_3^{X_{n+1}} \right)^{n+1-X_{n+1}}\right] \nonumber \\
        &= \mathbb{E}\left[\left(1 - c_3^{X_n + Y_{n+1}} \right)^{n+1-X_n - Y_{n+1}}\right] \nonumber \\
        &= (1-c_1)\mathbb{E}\left[\left(1 - c_3^{X_n} \right)^{n+1 - X_n}\right] + c_1\mathbb{E}\left[\left(1 - c_3^{X_n+1} \right)^{n - X_n}\right] \nonumber \\
        &\ge (1-c_1)\mathbb{E}\left[\left(1 - c_3^{X_n} \right)^{n + 1 - X_n}\right] + c_1S_n
    \end{align}
    From Hoeffding's inequality,
    $\mathbb{P}(X_n - nc_1 \le -n\epsilon) \le \exp(-2n\epsilon^2)$. Substituting $\epsilon = c_1/2$, we get 
    $\mathbb{P}(X_n \le nc_1/2) \le e^{-nc_1^2/2}$. This implies 
    \begin{align} \label{equation:suboptimal-minimizer-eq-2}
        &\quad\,\,S_n - \mathbb{E}\left[\left(1 - c_3^{X_n} \right)^{n + 1 - X_n}\right] \nonumber \\
        &= \mathbb{E}\left[\left(1 - c_3^{X_n} \right)^{n - X_n}\cdot c_3^{X_n}\right] \nonumber \\
        &\le \mathbb{E}\left[c_3^{X_n}\right] \nonumber\\
        &= \mathbb{E}[c_3^{X_n}\mathbbm{1}(X_n \le nc_1/2) + c_3^{X_n}\mathbbm{1}(X_n > nc_1/2)] \nonumber\\
        &\le 1\cdot \mathbb{P}(X_n \le nc_1/2) + c_3^{nc_1/2}\cdot \mathbb{P}(X_n > nc_1/2), \text{ by }c_3 \in (0, 1) \nonumber\\
        &\le e^{-nc_1^2/2} + c_3^{nc_1/2}
    \end{align}
    Combining equations (\ref{equation:suboptimal-minimizer-eq-1}) and (\ref{equation:suboptimal-minimizer-eq-2}) gives 
    \begin{align} \label{equation:suboptimal-minimizer-eq-3}
        S_{n+1} &\ge (1-c_1)\cdot \left(S_n - e^{-nc_1^2/2} - c_3^{nc_1/2} \right) + c_1S_n \nonumber \\
        &= S_n - (1-c_1)\left(e^{-nc_1^2/2} + c_3^{nc_1/2} \right) \nonumber \\
        &\ge S_n - 2b^n
    \end{align}
    where $b := \max\left(e^{-c_1^2/2}, c_3^{c_1/2} \right) \in (0, 1)$.
    From \cref{equation:suboptimal-minimizer-eq-3}, iterating gives 
    $$ S_m \ge S_n - 2(b^n + \cdots + b^{m-1}) \ge S_n - 2\frac{b^n}{1-b}$$
    for all $m > n$.
    Now let $p_n$ denote the group survival probability with $n$ regular members. 
    From equations (\ref{equation:homogeneous-group-immediate-contagion-eq-4}) and (\ref{equation:homogeneous-group-immediate-contagion-eq-5}), 
    $$p_m - p_n = \mathbb{P}(\eta_0 > T)\cdot (S_n - S_m) \le 2\frac{b^n}{1 - b}$$
    Hence, when $m > \lceil\log_b \frac{\delta (1-b)}{2} \rceil =: U$, we have $$p_m \le p_U + \delta.$$

    In other words, any group size larger than $U$ cannot improve the survival probability by more than $\delta$ compared to the group size $U$.
    Thus $U$ is guaranteed to be $\delta$-suboptimal, which is sufficient to establish our claim of this proposition.
\end{proof}

The propositions above show that to find a group size that is nearly optimal (or true optimal), it is not necessary to consider all $n \in \mathbb{N}$; instead, the search can be restricted to a finite range determined by the bounds above.

In \cref{figure:upper-bound} below, we demonstrate that for some values of $(c_1, c_3)$, the $\delta$-suboptimal upper bound is indeed smaller. We emphasize, however, that it is not guaranteed that the $\delta$-suboptimal bound is always the better one; in practice, one can simply take the minimum of the two bounds from Propositions \ref{prop: optimal-size-bound} and \ref{prop:suboptimal-minimizer}. 
The pseudo-code below shows how one can numerically determine a group size that yields a survival probability at most $\delta$-away from the true optimal value, i.e. $\text{surv. prob.}(n \text{ found below}) \ge \text{surv. prob.}(\text{true optimal }n) - \delta$.

\begin{verbatim}
    best_survival_prob <- 0
    U <- min(proposition 2, proposition 3)
    for n <- range(1, U):
        step 1:
            compute Sn via eq(13)
            compute survival probability via eq(12)
        step 2:
            if probability is higher than best_survival_prob:
                update best_survival_prob and record this n 
\end{verbatim}
The time cost for computing each $S_n$ is $\mathscr{O}(n)$ since we can use the recursive relation relation ${n \choose k+1} = {n\choose k}\cdot \frac{n-k}{k+1}$, so there is no need to compute binomial coefficients from scratch for different $k$'s. In total, the time complexity of the procedure is $\mathscr{O}(U^2)$, where $U$ is defined in the pseudo-code above.

\begin{figure}[h!]
    \centering
    \begin{subfigure}[b]{0.45\linewidth}
        \includegraphics[width=\linewidth]{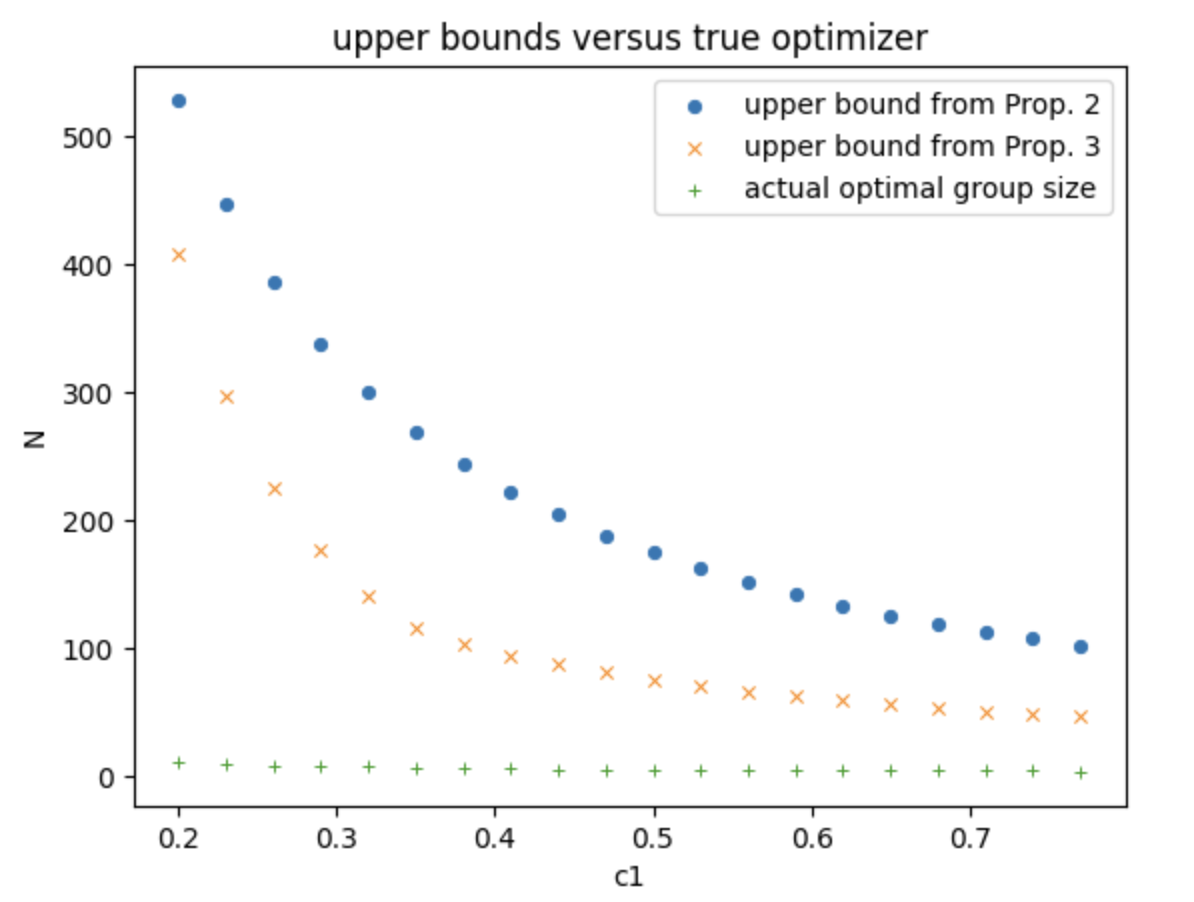}
        \caption{Contagion factor and suboptimal factor set to $q=0.3, \delta= 0.03$}
    \end{subfigure}
    \begin{subfigure}[b]{0.45\linewidth}
        \includegraphics[width=\linewidth]{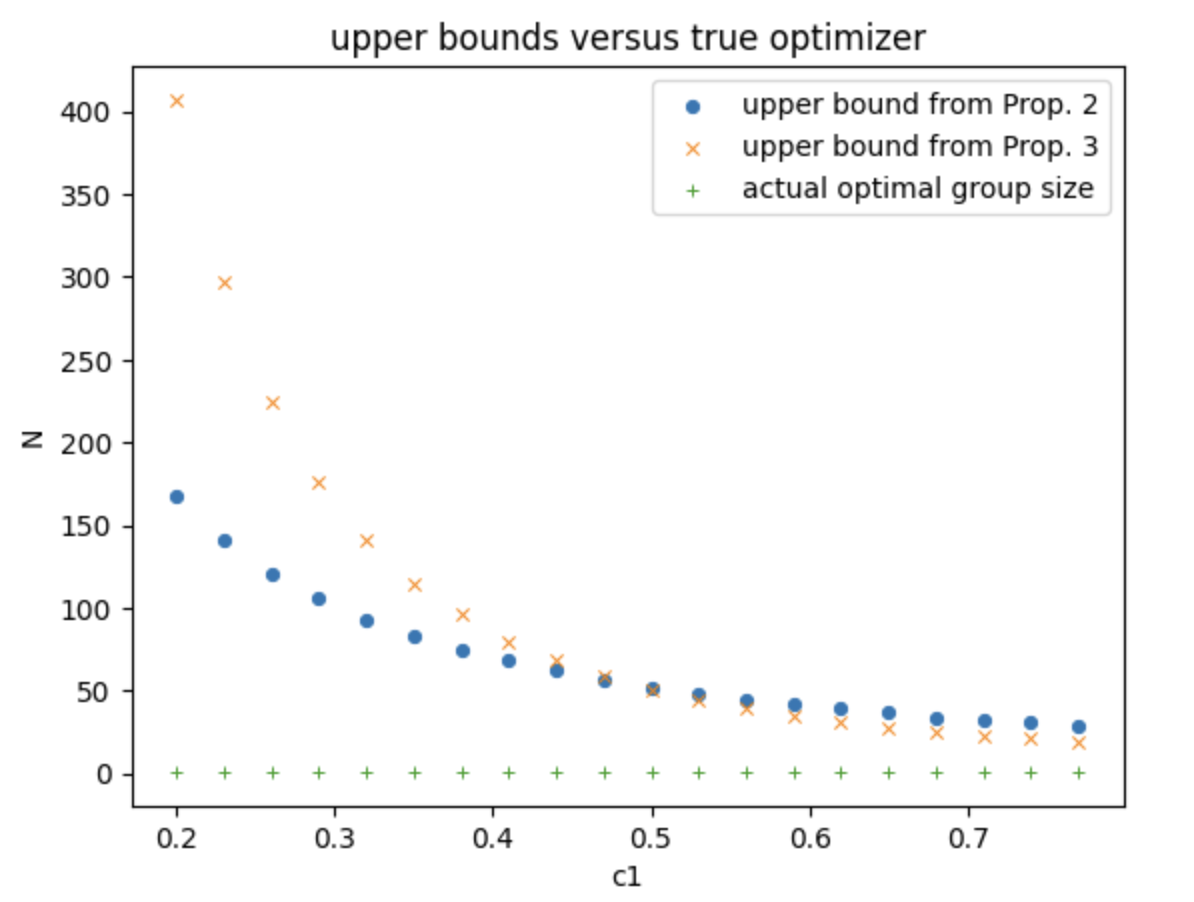}
        \caption{Contagion factor and suboptimal factor set to $q=0.6, \delta = 0.03$}
    \end{subfigure}
    \caption{Comparison between different upper bounds for the optimal group size}
    \label{figure:upper-bound}
\end{figure}

Observe in \cref{figure:upper-bound} that the true optimizer is usually much smaller than the upper bounds we provide (in the left plot in particular, the true optimizer varies from 10 to 4.)
We briefly explain the difficulty of obtaining a tight analytical upper bound for the optimal group size.
As mentioned above, if we substitute the random variable $X_n$ in \cref{equation:homogeneous-group-Sn-converges-eq-1} by its large number estimation $nc_1$, we can explicitly solve for its optimizer.
However, such an estimation is only accurate for large $n$'s, meaning we do not have control on the behavior of the true $S_n$ over small $n$'s.
Figure 2 above may lead us to think the group's survival probability curve follows the same general shape as we vary the group size; however, in some cases of $(c_1, c_3)$, we might obtain different patterns. We illustrate this with \cref{figure:Sn-pattern(a)}, where $c_1 = 0.01$ and $q = 0.77$.

\begin{figure}[h!]
    \centering
    \includegraphics[width=0.5\linewidth]{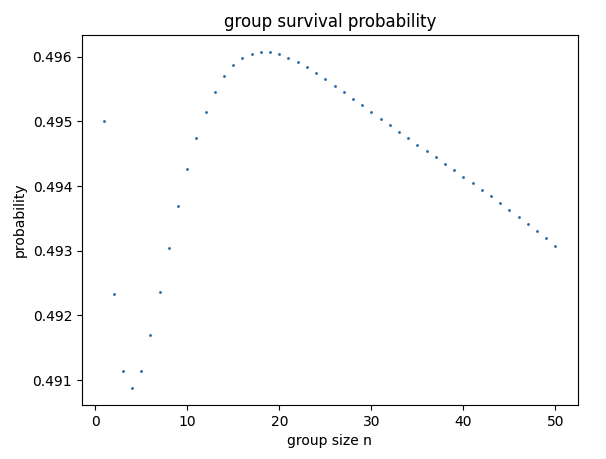}
    \caption{Another example of the group's survival probability with respect to $n$}
    \label{figure:Sn-pattern(a)}
\end{figure}

In addition, as we can see in \cref{figure:Sn-pattern(b)}, where $c_1 = 0.05$ and $q = 0.85$, the group's true survival probability differs noticeably from the estimated value (with $X_n$ substituted by $nc_1$) even until $n$ goes up to 100 (the true optimizer is at $n=1$, while the estimated curve peaks at $n = 7$.)

\begin{figure}[h!]
    \centering
    \includegraphics[width=0.5\linewidth]{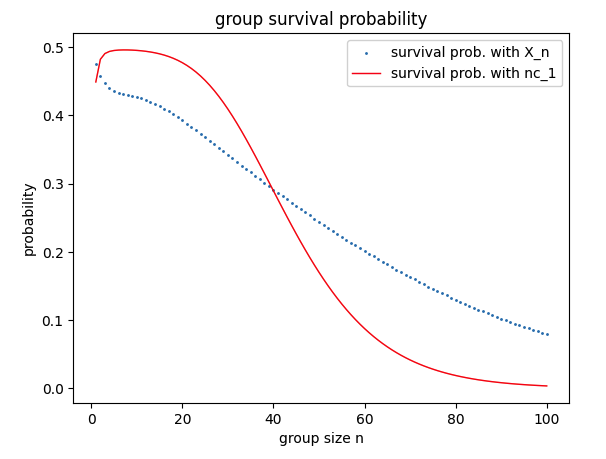}
    \caption{Compare the true group survival probability with its approximation, in which $X_n$ is replaced by $nc_1$}
    \label{figure:Sn-pattern(b)}
\end{figure}

As a summary, the plots here demonstrate the challenge of analyzing $S_n$ with small $n$'s.
We leave the task of finding a tighter upper bound of the optimizer as future work for anyone who is interested in improving it.


\section{Declarations}
\subsection{Funding}
Author Alejandra Quintos is supported by the Office of the Vice Chancellor for Research and Graduate Education at the University of Wisconsin–Madison, with funding from the Wisconsin Alumni Research Foundation.

\subsection{Competing interests}
The authors have no competing interests to declare that are relevant to the content of this article.

\printbibliography

\end{document}